\documentclass[11pt]{article}
\usepackage{fullpage}
\usepackage{epsfig}
\usepackage{multienum,subfigure}
\usepackage{algorithm}
\usepackage[noend]{algpseudocode}
\usepackage{amsmath,amssymb,amsthm}
\usepackage{ifsym}
\usepackage{verbatim}
\usepackage{color}
\usepackage{graphicx}
\usepackage{bbm}
\usepackage{hyperref}
\usepackage{todonotes}
\usepackage[square]{natbib}

\newtheorem{lemma}{Lemma}
\newtheorem{theorem}{Theorem}
\newtheorem{corollary}[theorem]{Corollary}

\newtheorem{proposition}[theorem]{Proposition}

\theoremstyle{definition}
\newtheorem{definition}{Definition}
\newtheorem{openp}{Open Problem}

\newtheorem{claim}{Claim}

\makeatletter
\def\blfootnote{\xdef\@thefnmark{}\@footnotetext}
\makeatother

\newcommand{\R}{\mathbb{R}}

\newcommand{\argmax}{\operatorname{arg\,max}}

\newcommand{\prob}[2][]{\text{\bf Pr}\ifthenelse{\not\equal{}{#1}}{_{#1}}{}\!\left[#2\right]}
\newcommand{\expect}[2][]{\text{\bf E}\ifthenelse{\not\equal{}{#1}}{_{#1}}{}\!\left[#2\right]}

\newcommand{\opt}{{\normalfont\textsc{Opt}}}
\newcommand{\alg}{{\normalfont\textsc{Alg}}}

\newcommand{\E}{\mathbb{E}}

\def\Pr{\ensuremath{\mathrm{Pr}}}

\def\argmax{\ensuremath{\mathrm{argmax}}}

\def\spann{\ensuremath{\mathrm{span}}}
\def\rank{\ensuremath{\mathrm{rank}}}
\newcommand{\ptope}{\mathcal{P}}
\newcommand{\optS}{S^*}
\newcommand{\algS}{\widehat{S}}
\DeclareMathOperator{\gout}{out}
\DeclareMathOperator{\gin}{in}
\DeclareMathOperator{\mspan}{span}
\newtheorem{fact}{Fact}
\def\I{\ensuremath{\mathcal{I}}}
\renewcommand{\S}{\mathcal{S}}

\newcommand{\notshow}[1]{{}}

\newcommand*{\pr}[2][]{\text{Pr}\ifx\\\left[#1\right]\\\else_{#1}\fi \left[#2\right]}

\defcitealias{CDW}{CDW '16}
\defcitealias{FGKK}{FGKK '16}
\defcitealias{DW}{DW '17}
\defcitealias{DHP}{DHP '17}
\defcitealias{Myerson}{Mye '81}
\defcitealias{DDT15}{DDT '15}
\defcitealias{SSW17}{SSW '18}
\defcitealias{MV}{MV '07}
\defcitealias{CheGale2000}{Che and Gale 2000}
\defcitealias{DGSSW}{DGSSW '19}
\defcitealias{CDGM19}{CDGM '19}
\defcitealias{EFFGK}{EFFGK '19}

\newcommand{\kira}[1]{\noindent{\textbf{\color{violet}KG: #1}} \index{KIRA}}
\begin{document}
\title{Non-Adaptive Matroid Prophet Inequalities}
\author{
Shuchi Chawla\\ University of Wisconsin-Madison\\ shuchi@cs.wisc.edu \and 
Kira Goldner\\ Columbia University\\ kgoldner@cs.columbia.edu \and 
Anna R. Karlin\\ University of Washington\\ karlin@cs.washington.edu \and 
J. Benjamin Miller\\ Google\\ bmiller@cs.wisc.edu
}
\date{\today}

\maketitle

\begin{abstract}

We investigate non-adaptive algorithms for matroid prophet inequalities.  Matroid prophet inequalities have been considered resolved since 2012 when [KW12] introduced thresholds that guarantee a tight 2-approximation to the prophet; however, this algorithm is adaptive.  Other approaches of [CHMS10] and [FSZ16] have used non-adaptive thresholds with a feasibility restriction; however, this translates to adaptively changing an item's threshold to infinity when it cannot be taken with respect to the additional feasibility constraint, hence the algorithm is not truly non-adaptive.  A major application of prophet inequalities is in auction design, where non-adaptive prices possess a significant advantage: they convert to order-oblivious posted pricings, and are essential for translating a prophet inequality into a truthful mechanism for multi-dimensional buyers.  The existing matroid prophet inequalities do not suffice for this application.  We present the first non-adaptive constant-factor prophet inequality for graphic matroids. 
\end{abstract}

\addtocounter{page}{-1}
\newpage

\thispagestyle{plain}

\section{Introduction} \label{sec:intro}

We study the classic prophet inequality problem introduced by \citet{KrengelS77}: $n$ items arrive online in adversarial order.  A gambler observes the value of each item as it arrives, and in that moment, must decide irrevocably whether to take the item or pass on it forever.  He can accept at most one item. The gambler knows in advance the (independent) prior distribution of each item's value.  What rule should he use to maximize the value of the item he accepts?  In expectation, how does the maximum value that the gambler can guarantee compare to the \emph{prophet}, who knows all of the realized item values in advance and selects the highest valued one?

The prophet inequality is a standard model for online decision making in a stochastic/Bayesian setting and has many applications, particularly to mechanism design and pricing. Over the last few years many variants of the basic single-item setting have been studied. One natural generalization is to allow the gambler to accept more than one item, subject to a feasibility constraint.
Formally, we can represent a feasibility constraint as a collection $\S$ of feasible sets.  
Then both the gambler and prophet can each select any feasible set of items $S \in \S$; in the single-item setting, the feasible sets are just all singletons.  What is the gambler's best algorithm and guarantee?


A seminal result by \citet{Samuel-Cahn}  showed that for the basic single-item setting the online algorithm can obtain at least half of the prophet's value in expectation by determining a single threshold $T$ and accepting the first item with value exceeding $T$. Further, this approximation factor is tight: there exist instances where the gambler can do no better than $\frac{1}{2}$ as well as the prophet.  The threshold $T$ is selected such that the probability that the value of \emph{any} of the $n$ items exceeds the threshold is exactly $\frac{1}{2}$.

In 2012\footnote{Their original result appeared in STOC 2012 \citep{KleinbergW12}, but we will cite their journal version from 2019 for the remainder of the paper.}, \citeauthor{KleinbergW12} introduced an alternative approach for setting a single threshold: set $T = \frac{1}{2} \opt$.  Here, $\opt$ is what the prophet can achieve, and this approach guarantees the same $\frac{1}{2}$-approximation for a single item.  
\citeauthor{KleinbergW12} showed that this alternate approach generalizes also to \emph{matroid} prophet inequalities: where both the gambler and the prophet are restricted to accepting independent sets in a given matroid. In this setting, the approach of \citeauthor{KleinbergW12} still achieves a factor of $2$ approximation, matching the single item lower bound.

There is a significant qualitative difference between Samuel-Cahn's approach for the single item prophet inequality and \citeauthor{KleinbergW12}'s approach for the matroid setting. In particular, the former computes a single threshold that is then used for the entire duration of the algorithm. The latter on the other hand, recomputes thresholds after every decision. The threshold applied to the value of the second item, for example, depends on whether the first item was accepted by the algorithm or not, and in turn on the realized value of the first item. As a consequence, the KW algorithm is more complicated and involves more computation.

In this paper we address a natural problem exposed by this discussion: {\bf Can an online algorithm compete against the prophet using static thresholds under a matroid feasibility constraint?} 


There is an inherent connection between prophet inequalities and Bayesian mechanism design.  The original problem by \citeauthor{KrengelS77} was formulated as an optimal stopping problem; it was \citet{HKS07} that made the first connection to an economic welfare-maximization problem.  \citet{CHMS} studied this connection much more deeply, defining a truthful class of simple mechanisms called ``order-oblivious posted pricings.''  
  They show that one can translate a prophet inequality for $n$ items with feasibility constraint $\S$ into an order-oblivious posted pricing for an $n$-unit setting with unit-demand buyers and a service feasibility constraint\footnote{A service feasibility constraint $\S$ says that the set of buyers that are served simultaneously must belong to some set $S \in \S$.} corresponding to $\S$; this mechanism is truthful and it yields a revenue guarantee that matches its prophet inequality guarantee.  
  When $\S$ is a matroid, by \citet{KleinbergW19}, the resulting mechanism yields $\frac{1}{2}$-approximation to the optimal expected revenue.

This reduction from truthful mechanisms to prophet inequalities crucially
relies on the buyers being unit-demand.  If we wish instead wish to translate the prophet inequality into a mechanism for a single constrained-additive buyer subject to feasibility constraint $\S$ over $n$ heterogenous items---that is, the buyer is interested in buying
\emph{more} than one item---then adaptive thresholds will not translate to a truthful
mechanism.  Instead, they correspond to offering each item one-at-a-time to the
buyer in any order, but prices change as a function of previous purchases.
This update will not generally preserve truthfullness; that is, although the
buyer may wish to purchase the first item offered when considered myopically,
he may be better off declining, in order to avoid price increases on later
items.

In order to fix this reduction for multi-parameter buyers beyond unit-demand, we must use only prophet inequalities with non-adaptive thresholds.  
This is our primary motivation: constructing non-adaptive prophet inequalities in order to expand the realm of settings where prophet inequalities can be used for truthful mechanism design (see related work for an understanding of how integral they are as a tool in this field).  However, non-adaptive prophet inequalities possess numerous other attractive properties as well.  For welfare-maximizing mechanisms, non-adaptive prophet inequalities correspond to prices that are not only order-oblivious, but also \emph{anonymous}, using the same prices on each item regardless of the buyer.  Additionally, since the thresholds (prices) are all computed before the items arrive and are never updated, there is much less computation required than for adaptive thresholds.




\subsection{Our Contribution and Roadmap}

We present the first non-adaptive thresholds that give a constant-factor prophet inequality for graphic matroids.  We finish Section~\ref{sec:intro} with additional related work and in Section~\ref{sec:prelim}, we introduce mathematical preliminaries.  In Section~\ref{sec:motiv}, we discuss why extending non-adaptive algorithms to graphic matroids is such a challenging objective, and why prior methods fail.  Section~\ref{sec:tipi} presents the ex-ante relation to the matroid polytope: a reduction from a given prophet inequality instance to an alternative setting with convenient properties for designing algorithms.  Expert readers can safely skip this section.  Then, in this context, Section~\ref{sec:graphic} presents our construction for non-adaptive thresholds.



The ex-ante relaxation takes a given item $i$'s value distribution and converts it into a Bernoulli distribution: with probability $p_i$, item $i$ is ``active,'' that is, non-zero, and takes on value $t_i$.  Then, the threshold for item $i$ is implicit (just the non-zero value $t_i$), and the only remaining questions are (1) with what probability should our algorithm consider this item, and (2) with what probability will the item be ``unblocked,'' or feasible to accept, when the gambler reaches it?

In order to obtain a constant-factor approximation, then the probabilities of both selection and feasibility must be constant.  In a graphic matroid, the elements are the edges and the independent sets are the forests---that is, any set of edges that does not contain a cycle.  Depending on the given graph, an edge could be ``blocked'' by many different edges, and it is unclear how to group elements.  Our main idea is to orient the graph to have good properties and then exploit them.  For an edge $(u,v)$, suppose it is oriented into vertex $v$.  Notice that if no other edges incident to $v$ are selected by the algorithm, then $(u,v)$ will certainly be feasible to accept---it cannot possibly form a cycle by taking this edge.  We orient the graph such that all edges directed into $v$ have low enough probability mass such that with probability $\frac{1}{2}$, none are active.  Then, our algorithm decides which edges to consider such that, with constant probability for every $v$, it will consider the edges into $v$ and \emph{not} the edges out of $v$.  Hence, with no edges out of $v$ and a good chance that no edges into $v$ will be active, any edge into $v$ can be accepted with constant probability.  Our method for determining which edges to consider is simple: we take a random cut and consider only the edges in one direction across the cut.  Since every edge is oriented into some vertex $v$, it will be both considered and unblocked with constant probability, as desired.


Unfortunately, this approach is quite specific to a graphic matroid.  While some properties of the algorithm might extend to other matroids, we know that it cannot generalize to all matroids: \citet{FeldmanSZ19} prove a lower bound of $\Omega(\frac{\log n}{\log \log n})$ for prophet inequalities that use only non-adaptive thresholds for the class of general matroids.  Their lower bound example is a gammoid.

We pose the following two remaining open (but likely very difficult) open questions for understanding how far non-adaptive constant-factor approximations reach between graphic matroids and the lower bound of a gammoid.

\begin{openp}
What is the boundary within matroids for non-adaptive constant-factor approximations?
\end{openp}

\begin{openp}
How do approximations decay for non-adaptive thresholds as matroids become more complex?
\end{openp}




\subsection{Additional Related Work} \label{sec:related}

\paragraph{Non-Adaptive Thresholds.} As mentioned, the two predominant
approaches for achieving $\frac{1}{2}$-approximation in the single-item setting
are both non-adaptive~\citep{Samuel-Cahn,KleinbergW19}.  \citet{CHMS} provide
non-adaptive $\frac{1}{2}$-approximations to the prophet for both $k$-uniform
and partition matroids; they also give a non-adaptive $O(\log r)$-approximation
for general matroids, where $r$ is the rank of the matroid.  Recent work by
\citet{GravinW19} gives a non-adaptive algorithm that guarantees a
$3$-approximation to the prophet for online bipartite matching, which is the
intersection of two matroids.  \citet{ChawlaDL20} optimize non-adaptive
thresholds for the $k$-uniform settings depending on the range that $k$ is in,
improving existing guarantees in the $k < 20$ regime.  No non-adaptive
algorithms are known beyond uniform and partition matroids and the special case
of bipartite matching.

\paragraph{Constrained Non-Adaptive.} Another class of algorithms uses non-adaptive thresholds and a restricted feasibility constraint.  That is, given a feasibility constraint $\S$ and the prior distributions for $n$ items, the algorithms set, prior to the arrival of all items, thresholds $T_i$ for each item and a restricted feasibility constraint $\S'$ such that $\S' \subset \S$.  Then, an item is accepted if it exceeds its threshold \emph{and} is feasible with respect to the items already accepted and the subconstraint $\S'$.  Notice that an item could exceed its threshold, be feasible with respect to previously accepted items and $\S$, and yet not be accepted because it is not feasible with respect to previously accepted items and $\S'$.  In essence, imposing a subconstraint is equivalent to adaptively changing an item's threshold to $T_i = \infty$ if the item is not feasible with respect to the subconstraint.  


Why is this different than when the gambler rejects an item that exceeds its threshold but is not feasible with respect to $\S$?  We can interpret the gambler's value as constrained-additive with respect to $\S$, so the gambler does not have any marginal gain for items that are infeasible with respect to $\S$ and the items he has already accepted.  Hence, he has no reason to take items with no positive marginal value to him.  This is not a restriction on the algorithm, but rather a result of the gambler's valuation class.

\citet{CHMS} first produced a $\frac{1}{3}$-approximation to the prophet for graphic matroids using non-adaptive thresholds with a partition matroid subconstraint.  In a very elegant approach, \citet{FeldmanSZ16} produce an Online Contention Resolution Scheme (OCRS) that yields a $\frac{1}{4}$-approximation for all matroids using non-adaptive thresholds and a subconstraint built cleverly from the structure of the given matroid.


\paragraph{Prophet Inequalities Beyond Matroids.} Prophet inequalities are well-studied and the literature is far too broad to cover; see \citep{Lucier17} for an excellent survey.  Note, however, that dynamic algorithms yield good approximations to the prophet in settings reaching beyond matroids.  In addition to matroids, the approach of \citep{FeldmanSZ16} also applies to matchings, knapsack constraints, and the intersections of each.  Very recent work by \citet{DuttingKL20} gives an algorithm guaranteeing an $O(\log \log n)$-approximation for the very general setting of multiple buyers with subadditive valuations.

\paragraph{Direct Applications to Pricing.}  The \citet{CHMS} reduction from order-oblivious posted pricings to prophet inequalities was only the first of many pricing applications of prophet inequalities.  
\citet{FGL} considers the setting where buyers arrive online and face posted prices for items; non-adaptive anonymous prices are posted for each item equal to half its contribution to the optimal welfare.  These prices guarantee $\frac{1}{2}$-approximation to the optimal welfare for fractionally subadditive valuations.  Note that this is a prophet inequality when there is only one item.  
\citet{DuttingFKL17,DuttingFKL20} connect posted prices and prophet inequalities: they interpret the \citet{KleinbergW19} thresholds as ``balanced prices'' and derive an economic intuition for the proof.  They extended these balanced prices to more complex settings, including a variety of feasibility constraints and valuation classes.  The approach is to prove guarantees in the full information setting, where the realized values are known in advance.  Then, via an extension theorem, they prove that the results hold for Bayesian settings too, where distributions are known but values are unknown.  Note that their balanced prices result in non-adaptive anonymous prices for all settings they consider \emph{except} for matroids feasibility constraints, where they remain adaptive and buyer-specific.
The recent work of \citet{DuttingKL20} also implements posted prices for buyers with subadditive valuations, but rather than balanced prices, provides a weaker sufficient condition to get a tighter approximation, and shows the existence of such prices through a primal-dual approach.



\paragraph{More Subtle Applications in Mechanism Design and Analysis.} Beyond direct applications to pricing, prophet inequalities have also been used in to build more complex mechanisms and prove approximation guarantees.  \citet{CM} design a two-part tariff mechanism to approximate optimal revenue for matroid-constrained buyers.  Their benchmark is an ex-ante relaxation, and they use an OCRS \citep{FeldmanSZ16} to achieve a constant fraction of that revenue.  \citet{CZ17} prove that the better of a sequential posted price mechanism (where each buyer can only buy one item) and an anonymous sequential posted price mechanism with an entry fee yields a constant-approximation to the optimal revenue for multiple fractionally subadditive buyers (and $O(\log n)$-approximation for fully subadditive). In a specific case of their analysis that analyzes the core of the core (a double core-tail analysis follow the original of \citep{LY13}), they use \citep{FGL}.  
Work by \citet{CZ19} approximates the optimal profit---seller revenue minus cost---for constrained-additive buyers.  Like \citep{CM}, they also construct their benchmark using the ex-ante relaxation and use OCRS to bound a term here as well.  
Recent work by \citet{CaiGMZ20}  studies gains from trade approximation in a two-sided market with a constrained-additive buyer and single-dimensional sellers---both the single-item prophet inequality of \citep{KleinbergW19} and an OCRS are used to inspire prices for \emph{both} the buyer and the sellers simultaneously and then show that enough gains from trade will be received to approximate one specific part of their benchmark.

\section{Preliminaries} \label{sec:prelim}

\begin{definition}
A \emph{matroid} $M = (N, \I)$ is defined by a ground set of elements $N$ (with $|N| = n$) and a set of independent sets $\I \subseteq 2^N$.  It is a matroid if and only if it satisfies the following two properties:
\begin{enumerate}
\item Downward-closed: If $I \subset J$ and $J \in \I$ then $I \in \I$.
\item Matroid-exchange: For $I, J \in \I$, if $|J| > |I|$ then there exists some $i \in J \setminus I$ such that $I \cup \{i\} \in \I$.
\end{enumerate} \end{definition}

\noindent We review several standard notions for matroids:
\begin{itemize}
\item The \emph{rank} of a set $\rank(S)$ is the size of the largest independent set in $S$: $\max \{|I| \mid I \in \I, I \subseteq S\}$.
\item The \emph{span} of a set $\spann(S)$ is the largest set that contains $S$ and has the same rank as $S$: $\{i \in N \mid \rank(S \cup \{e\}) = \rank(S)\}$.
\item An element $i$ is \emph{spanned} by a set $S$ when $i \in \spann(S)$.
\end{itemize}

We will informally use the language ``blocked'' (by a set $S$) to mean that an element is spanned (by the set $S$), and similarly ``unblocked'' to mean that an element is \emph{not} spanned (by the set $S$).

For any matroid $M$, we have the \emph{matroid polytope} $\ptope_M = $ 
$\{ \vec{p} \in \R_{\geq 0} ^M \mid \forall S \in 2^N, \sum _{i \in S} p_i \leq \rank(S) \}$.  That is, $\ptope_M$ is the convex hull of the independent sets $\I$.

\begin{definition} A \emph{Matroid Prophet Inequality instance} $(\vec{X}, M)$ is given by a matroid $M =(N, \I)$ and distribution of values $\vec{X}$ for the $n$ items that are the ground set $N$.  $X_i$ denotes the random variable representing the value for item $i$.  \end{definition}

For any given matroid prophet inequality instance, we let $\opt(\vec{X}, M)$ denote the value of the prophet's set in expectation of the value of the items.  Formally, $\opt(\vec{X}, M) = \E\big[\max_{I\in \I} \sum_{i\in I}X_i\big]$.  We omit the distributions $\vec{X}$ or matroid $M$ when it is obvious from context.

\begin{definition}A \emph{non-adaptive} threshold algorithm is given an instance $(\vec{X}, M)$ and determines thresholds $\vec{T}$.  A threshold $T_i$ for each item $i$ is a function only of the random variables $\vec{X}$ (and, in particular, not as a function of any realizations of $\vec{X}$ or whether previous items have exceeded thresholds thus far). \end{definition}

For any non-adaptive thresholds $\vec{T}$, we let $\alg(\vec{X}, M, \vec{T})$ denote the expected value obtained by the algorithm.  Again, we omit the parameters when they are clear from context.
\section{Where Straightforward Extensions Fail} \label{sec:motiv}

Both of the non-adaptive single-item approaches---the probabilistic approach of
\citet{Samuel-Cahn} and the $\frac{1}{2} \opt$ approach of
\citet{KleinbergW19}---extend to the $k$-uniform matroid setting, in which any
set of size at most $k$ is feasible.  We first see why these approaches work for
$k$-uniform matroids yet break down for graphic matroids.  Then, we attempt to
use an idea for graphic matroids from \citet{CHMS} to develop a non-adaptive
algorithm, and again highlight where the approach breaks down.

We begin with the two generalizations to $k$-uniform methods.  Note that we do not claim either as part of our contribution, although to the best of our knowledge, neither approaches' generalized thresholds and proof is written anywhere.



Formally, a $k$-uniform matroid is the matroid where, for any given ground set $N$, $\I = \{I \subset N : |I| \leq k\}$.  Bear in mind that $k=1$ returns to the single-item case.

\paragraph{The Probabilistic Approach.} (Extension of \citet{Samuel-Cahn} single-item algorithm to non-adaptive thresholds for the $k$-uniform matroid.)  Determine the thresholds $T$ by setting $\Pr[\text{$< k$ item values exceed $T$}] = \Pr[\text{$\geq 1$ slot empty}] = p = \frac{1}{2}$.  

\begin{align*}
\alg(\vec{X}, T) &\geq \sum _i \Pr[i \text{ not blocked}] \E[(X_i - T)^+] + \Pr[\text{$\geq k$ above T}] \cdot k T \\
&\geq \Pr[\text{$< k$ above T}] \sum _i  \E[(X_i - T)^+] + \Pr[\text{$\geq k$ above T}] \cdot k T \\
&\geq p  \E \left[ \max_{S: |S| \leq k} \sum_{i \in S} (X_i - T)^+ \right] + (1-p) k T \\
&\geq p  \E\left[\max_{S: |S| \leq k} \sum_{i \in S} X_i - kT\right] + (1-p) kT \\
&= \frac{1}{2} ( \E\left[\max_{S: |S| \leq k} \sum_{i \in S} X_i \right] ) - \frac{1}{2} kT + \frac{1}{2} kT\\
&=  \frac{1}{2} \E\left[\max_{S: |S| \leq k} \sum_{i \in S} X_i \right] = \frac{1}{2} \opt(\vec{X}).
\end{align*}

For uniform matroids, a simple characterization based on size exists
for sets that do not span \emph{any} elements that have yet to arrive: they
need only be of size strictly less than $k$.  This property does not hold for
more complex matroids.

\paragraph{The ``Thresholds as Constant-Fraction of Prophet'' Approach.}  (Extension of \citet{KleinbergW19} single-item algorithm to non-adaptive thresholds for the $k$-uniform matroid; almost identical to those in \citet{CHMS}.)  Set $T = \frac{1}{2k} \E\left[\max_{S: |S| \leq k} \sum_{i \in S} X_i \right] = \frac{1}{2k} \opt(\vec{X})$.

\begin{align*}
\alg &\geq \sum _i \Pr[i \text{ not blocked}] \E[(X_i - T)^+] + \Pr[\text{$\geq k$ above T}] kT \\
&\geq \Pr[\text{$< k$ above T}]  \E[\sum _i (X_i - T)^+] + \Pr[\text{$\geq k$ above T}] kT \\
&\geq p  \E\left[\max_{S: |S| \leq k} \sum_{i \in S} (X_i - T)^+\right] + (1-p) kT \\
&= p \left( \E\left[\max_{S: |S| \leq k} \sum_{i \in S} X_i \right] - kT\right) + (1-p) \frac{1}{2} \E\left[\max_{S: |S| \leq k}\sum_{i \in S} X_i \right] \\
&=  \frac{1}{2} \E\left[\max_{S: |S| \leq k} \sum_{i \in S} X_i \right] = \frac{1}{2} \opt(\vec{X}).
\end{align*}

In uniform matroids, any element is exchangeable for any other element.  Then
so long as it contributes enough value, such as at least a constant fraction of
the average contribution to the optimal basis, there is no reason not to
accept an element.  However, this does not hold for more complex matroids.  A
particular element, even if extremely high value, may cause so many other
elements to be spanned that it is not worth taking.
  
  One can imagine more nuanced extensions of either such approach---probabilistic thresholds for $i$ according to how many elements it might block, or value-based thresholds for $i$ based on the value of the sets it might block.  However, any such extension would require a matroid-specific understanding of the relationship between elements, and element-specific thresholds.
  
  Note that in addition to uniform matroids, both approaches easily extend to partition matroids by applying the approach to thresholds specific to the uniform matroid in each partition.

\paragraph{The Constrained Non-Adaptive Approach.}  \citet{CHMS} construct non-adaptive thresholds for a graphic matroid that work \emph{so long as} the algorithm can enforce an additional subconstraint.  Specifically, they cleverly partition the graph such that, so long as at most one edge is accepted from each partition, then an independent set is guaranteed.  Then as items arrive, they are accepted if and only if they exceed their threshold \emph{and} are feasible with respect to the subconstraint---that is, no previous item from its partition has been accepted.  This approach guarantees a $\frac{1}{3}$-approximation.

As discussed in the introduction, enforcing a subconstraint \emph{is} in fact adaptive.  But, we \emph{could}, for example, randomly select one item from each partition in advance, defining our set for consideration $C$.  Then, as items arrive, in each partition, we consider only the item in $C$, ignoring all other items from each partition.  That is, we leave thresholds the same for all items in $C$ and \emph{a priori} set $T_i = \infty$ for all $i \not \in C$.  This ensures that we only consider a set that complies with our feasibility constraint \emph{without} making any modifications online.  Note that we can select items to be in the consideration set $C$ with whatever probabilities we choose, even in a correlated fashion---as long as we make them prior to items arriving---thus setting all thresholds to $T_i$ or $\infty$ in advance.  Is there some clever way that we can implement our feasibility constraint, or any feasibility constraint, yet maintain a constant-factor approximation?

For the approach of CHMS, we might observe that a convenient property that bounds the probability mass of each partition could allow us to form a probability distribution over elements in each partition (i.e. place item $i$ in $C$ with probability $p_i/2$).  However, this approach in fact reduces the probability too much, as it combines the probability that the element is active with the probability it is considered, and is no longer constant.  If we use a constant probability, it would instead sell to too low of a quantile.

If such an approach \emph{were} to work, we could convert \emph{any} non-adaptive matroid prophet inequality to a prophet inequality, as a greedy OCRS exists for all matroids and constructs constrained non-adaptive thresholds all matroids \citep{FeldmanSZ16}.  However, \citet{FeldmanSZ19} also prove a super-constant lower bound of $\Omega(\frac{\log n}{\log \log n} )$, so guarantees cannot possibly go through for every matroid.  Thus, an interesting direction for future work is to characterize \emph{when} an approach of converting constrained non-adaptive thresholds to a fully non-adaptive algorithm in this way would maintain good guarantees.

\section{The Ex-Ante Relaxation to the Matroid Polytope}
\label{sec:tipi}

Reducing a given matroid prophet inequality instance to one with Bernoulli distributions that sits within the matroid polytope is ``standard,'' and is used in \citep{FeldmanSZ16}.  It's ``just'' an ex-ante relaxation to the matroid polytope, and expert readers can safely skip this section.  However, we present the reduction in detail for comprehensiveness and ease of reading, as we did not find it elsewhere.   

First, given arbitrary independent random variables $X_i$, we reduce the problem to
designing an algorithm for independent Bernoulli random variables
$X'_i$:
\begin{align*}
    X_i' = \begin{cases}
        t_i & \text{w.p. } p_i \\
        0 & \text{w.p. } 1-p_i,
    \end{cases}
\end{align*}
where $\vec{p} \in \ptope_M$.

Reducing to Bernoulli random variables gives two properties which greatly
simplify the design of an algorithm:
\begin{enumerate}
    \item Each element of the ground set is either {\it active} or {\it
        inactive}; and
    \item There exists a worst-case total ordering of the elements.
\end{enumerate}
The worst-case ordering is the typical greedy ordering. Assume $t_i \le
t_{i+1}$; then greedily selecting elements in order (maintaining independence
and according to the rules of our algorithm) results in the lowest
weight outcome over all orderings. For the rest of the paper, we assume $t_i
\le t_{i+1}$ for $1 \le i < n$.

We now state our reduction formally.
\begin{lemma}
    \label{lem:reduction}
    Given a matroid $M = (N, \I)$ and independent random weights $X_i$, $i \in
    N$, there exist independent Bernoulli weights $X'_i$, where $X'_i = t_i$
    w.p.  $p_i$ and $\vec{p} \in \ptope_M$, such that 
    \[
        \opt(\vec{X}, M) \le \sum_ip_it_i.
    \]
    Furthermore, for any algorithm $\alg$,
    \[
        \alg(\vec{X}) \ge \alg(\vec{X'}).
    \]
\end{lemma}

\begin{proof}
    First, rewrite the original optimal value as a sum over the ground set:
    \begin{align*}
        \opt(\vec{X}, M) &= \E\big[\max_{I\in \I} \sum_{i\in I}X_i\big] \\
                         &= \sum_{i\in N} \Pr[i\in I^*] \cdot \E[X_i \,|\, i \in I^*],
    \end{align*}
    where $I^*$ is the maximum weight basis: $I^* = \argmax_{I \in \I}\sum_{i\in
    \I} X_i$. Now let $p_i = \Pr[i\in I^*]$---the ex-ante probability that $i$ is in the prophet's solution. Since $\vec{p}$ is a convex
    combination of basis vectors, then $\vec{p} \in \ptope_M$.

    Now, observe that $\E[X_i\,|\,X_i\ge F_i^{-1}(1-p_i)] \ge \E[X_i | H]$ for
    any event $H$ with $\Pr[H] = p_i$. Let $t_i = \E[X_i\,|\,X_i\ge
    F_i^{-1}(1-p_i)]$; then in particular $t_i \ge \E[X_i\,|\,i\in I^*]$. Hence
    \[
        \opt(\vec{X}, M) \le \sum_ip_it_i.
    \]

    Finally, to see that $\alg(\vec{X}) \ge \alg(\vec{X'})$, we simply couple
    $\vec{X}$ and $\vec{X'}$, so that $X_i \ge t_i$ if and only if $X'_i = t_i$. For any
    ordering of the elements, the algorithm applied to the original instance
    selects the same items as the algorithm applied to the Bernoulli instance.
\end{proof}

\section{A Constant-Factor Approximation for Graphic Matroids}
\label{sec:graphic}

Given a Bernoulli instance from the matroid polytope, we show how to utilize it to obtain a constant-factor
non-adaptive algorithm for graphic matroids.

A graphic matroid is defined by an undirected graph $G$ with vertices $V$ and
edges $E$. The edges of the graph form the ground set, and the independent sets
$\I$ are forests, i.e. cycle-free sets of edges: $\I = \{I \subseteq E : I
\text{ contains no cycles}\}$; every spanning tree is a basis.  In light of Lemma~\ref{lem:reduction}, each edge $i \in E$
has an associated weight $t_i$ and is active (non-zero) with probability $p_i$, where $\vec{p}\in\ptope_G$, the matroid polytope for the graphic matroid $G$. The objective is then to select a maximum weight
spanning tree.  
As discussed in the previous section, we assume the edges arrive in order with
$t_i \le t_{i+1}$ for all $1 \leq i \leq n-1$; this order obtains the worst-case performance.


Our approach works by considering only a subset of the edges which has the properties that (1) a
significant fraction of the prophet's benchmark is accounted for and yet (2) with constant probability, 
elements selected earlier in the ordering do not block later elements.

Specifically, we do this in two steps. First, we show there exists a way to
direct the edges such that every edge has at most a constant probability of
being spanned by edges {\it except} for those leaving the vertex into which it is directed. Then, we take a random
cut in the graph and allow our algorithm to select only edges crossing the cut
in one direction, ensuring that for every vertex, the edges entering it are considered while the edges leaving it are not with constant probability.

\paragraph{Notation.} We use $b_i(S)$ to denote the probability that element
$i$ is ``blocked'' or \emph{spanned} by the active elements in a set $S$ with
respective to active probabilities $\vec{p} \in \ptope_M$.  For $\vec{p} \in
\ptope_M$, let $R_{\vec{p}}(S)$ be the random set containing $i \in S$
independently with probability $p_i$.  We call this the ``active'' set.
Formally, $b_i(S) = \Pr[i \in \mspan(R_{\vec{p}}(S\setminus \{i\}))]$.  Notice
that even if $i \in S$, we do not worry that it would span itself.

One convenience of using the ex-ante relaxation is that, so long as each
element is unblocked with constant probability, that is, $1 - b_i(S) \geq c$,
we obtain a constant-factor approximation.

\subsection{Directing the Graph}

\begin{lemma}
\label{lem:direction}
For $p \in \frac14\ptope_G$, there exists a way to orient the edges of $G$ such
that for each vertex the total probability mass of incoming edges is at most $1/2$.
\end{lemma}

\begin{proof}
Any vector from the graphic matroid polytope $\mathcal{P}_G$ is a convex combination bases, or spanning trees.  The average vertex degree in any spanning tree is at most 2, so the average fractional degree in a convex combination of spanning trees is at most 2, and hence the average fractional degree under the scaled $\vec{p} \in \frac{1}{4} \mathcal{P}_G$ is at most $\frac{1}{2}$.  

Let in-deg$(v)$ denote the fractional in-degree of $v$ in the constructed directed graph.  That is, the sum of the ``active'' probabilities for the edges directed into $v$.  We can find an orientation of the edges in the graph given probabilities $\vec{p}$ such that in-deg$(v) \leq \frac{1}{2}$ for all vertices $v$: because the average degree is at most $\frac{1}{2}$, there exists some vertex $v$ with degree at most $\frac{1}{2}$.  Orient all of the edges incident to $v$ toward $v$, as in-deg$(v) \leq \frac{1}{2}$, and then recurse on the graph among the remaining vertices.
\end{proof}

\begin{corollary}
\label{cor:inBlockage}
Given a graph as guaranteed by Lemma~\ref{lem:direction}, let $\gin(v)$ be the
set of incoming edges to vertex $v$ and let $\gout(v)$ be the outgoing edges. For
any $i$, let $v$ be the vertex such that $i \in \gin(v)$. Then for any $S
\subseteq E$,
\[
    b_i(S\setminus\gout(v)) \leq \frac12 .
\]
\end{corollary}
\begin{proof}
Observe that for $i \in \gin(v)$, $v$ cannot be spanned by a set that contains no other edges incident to $v$.  Then in order for $i$ to be spanned in $S\setminus\gout(v)$, at least
one edge in $\gin(v)$ other than $i$ must be active.  By construction,
$\sum_{i\in\gin(v)}p_i \leq \frac12$. So the probability that no edges are
active is at least $\frac12$ by the union bound.
\end{proof}

\subsection{Random Cut}

Assume $\vec{p} \in\frac14\ptope_G$, and direct the graph as described above. Let
$A\subseteq V$ be a random set of vertices such that each vertex is included in
$A$ independently with probability 1/2, and let $\bar{A} = B = V \setminus A$.  Let  $\algS$ be the set of directed edges across the cut from $A$ to $B$, formally, $\algS = \{i: i \in \gout(u) \cap \gin(v), u \in A, v \in B \}$.  
\begin{claim}
\label{claim:randomS}
\[
\E_{\algS}\left[\sum_{i\in\algS}p_it_i(1-b_i(\algS))\right] \geq
\frac18\sum_{i\in E}p_it_i.
\]
\end{claim}

\begin{proof}
\begin{align*}
  \E_{\algS}\left[\sum_{i\in\algS}p_it_i(1-b_i(\algS))\right]
    &= \sum_{(u,v)\in E}p_{uv}t_{uv}\, \Pr[(u,v) \in \algS] \, 
            \E\left[1-b_{uv}(\algS)\middle| (u,v) \in \algS \right] \\  
    &= \sum_{(u,v)\in E}p_{uv}t_{uv}\, \Pr[u\in A] \Pr[v \in B]\, 
            \E\left[1-b_{uv}(\algS)\middle| u\in A, v \in B \right] \\
    &= \frac{1}{4} \sum_{(u,v)\in E}p_{uv}t_{uv} \, 
            \E\left[1-b_{uv}(\algS)\middle| u\in A, v \in B \right] \\            
    &\geq \frac{1}{8} \sum_{(u,v)\in E}p_{uv}t_{uv} \\                  
\end{align*}
where the last inequality follows from Corollary~\ref{cor:inBlockage}.
\end{proof}

\subsection{Final Algorithm}

For discrete random variables $\vec{X}$, our algorithm is constructive, albeit not
efficient, because we can compute $\vec{p}$ and $\vec{t}$ as guaranteed by
Lemma~\ref{lem:reduction}. (Of course, we can discretize continuous random variables to
arbitrary approximation.)

\begin{algorithm}[ht!]
  \begin{algorithmic}[1]
      \State Compute $\vec{p}$ and $\vec{t}$ as guaranteed by
      Lemma~\ref{lem:reduction}.
      \State Direct the graph as outlined in Lemma~\ref{lem:direction}.
      \State Choose a cut $(A,B)$ uniformly at random; let $\algS = \{i: i \in \gout(u) \cap \gin(v), u \in A, v \in B \}$.
      \State For all edges $i \in \algS$, set $T_i = t_i$.
      \State For all edges $i \not\in \algS$, set $T_i = \infty$.
  \end{algorithmic}
\end{algorithm}

Step 3 can be derandomized using the standard Max-Cut derandomization.
Our main result is that this algorithm gives a $\frac{1}{32}$-approximation.

\begin{theorem}
    Let $G$ be a graphic matroid with independent edge weights $\vec{X}$.
    Then 
    \[
        32\,\E[\alg(G,\vec{X})] \ge \opt(G,\vec{X}).
    \]
\end{theorem}

\begin{proof}
    Let $\vec{p}\in \ptope_G$ and $\vec{t}$ be the probabilities and values
    guaranteed by Lemma~\ref{lem:reduction}.  Let $p'_i = \frac14 p_i$.  Then our algorithm obtains  $\alg =
    \E_{\algS}\left[\sum_{i\in\algS}p_i't_i(1-b_i(\algS))\right]$, which by our construction of $\algS$ and Claim~\ref{claim:randomS}, gives
    $$\alg
    \geq \frac18\sum_{(u,v)\in E}p'_{uv}t_{uv} = \frac{1}{32} \sum _{i \in E}
    p_i t_i.$$
\end{proof}

Our approximation factor is, of course, a factor of 16 worse than the dynamic thresholds of \citep{KleinbergW19} and a 10.67-factor worse than the constrained non-adaptive thresholds of \citep{CHMS}.  However, our guarantee holds for fully non-adaptive thresholds, and thus will guarantee truthful mechanisms in multi-parameter mechanism design applications.

\bibliographystyle{plainnat}   
\bibliography{masterbib.bib}


\end{document}